\documentclass[11pt]{article}
\usepackage[T1]{fontenc}
\usepackage{amsmath,amsthm,amssymb,mathtools}
\usepackage[hmargin=1in,vmargin=1in]{geometry}
\usepackage[hidelinks]{hyperref}
\usepackage{microtype}
\hypersetup{pdftitle={Parity Tests with Ties}}

\newtheorem{theorem}{Theorem}
\newtheorem{lemma}{Lemma}
\newtheorem{corollary}{Corollary}
\newtheorem{proposition}{Proposition}
\newtheorem{remark}{Remark}

\newcommand{\R}{\mathbb{R}}
\newcommand{\ceil}[1]{\left\lceil #1\right\rceil}
\newcommand{\floor}[1]{\left\lfloor #1\right\rfloor}
\newcommand{\card}[1]{\left|#1\right|}
\DeclareMathOperator{\sign}{sign}

\title{Parity Tests with Ties}
\author{Ron Kupfer}
\date{}

\begin{document}
\maketitle

\begin{abstract}
We extend the Ting--Yao randomized maximum-finding algorithm~\cite{TY94}
to inputs that need not be pairwise distinct: each parity test
$P(i,B)=\prod_{a\in B}(x_i-x_a):0$ on $B\subseteq[n]\setminus\{i\}$ is simulated by
$O(\log\card{B})$ ordinary polynomial tests, raising depth from
$O((\log n)^2)$ to $O((\log n)^3)$ while preserving the $O(n^{-c})$ failure
probability for every fixed $c>0$.
\end{abstract}

\section{Introduction}\label{sec:intro}

A polynomial decision tree is a decision tree using ternary tests of the form
$p(x_1, x_2, \dots, x_n) : 0$ where $p$ are polynomials over the reals.
Ting and Yao~\cite{TY94} prove that for every fixed $c>0$ there is a randomized
polynomial decision tree of depth $O((\log n)^2)$ that, on pairwise distinct
inputs $x_1,\dots,x_n\in\R$, finds a maximum with failure probability
$O(n^{-c})$.  Internal nodes are \emph{parity tests}: for some pivot
$i\in[n]$ and $B\subseteq[n]\setminus\{i\}$,
\[
  P(i,B)\;:=\;\prod_{a\in B}(x_i-x_a)
\]
is compared to~$0$.  Under distinctness every factor is nonzero, the outcome is
strict, and
\[
  \sign P(i,B)=(-1)^{\card{\{a\in B:x_a>x_i\}}}.
\]

If repeated values are allowed, some $a\in B$ may satisfy $x_a=x_i$, so
$P(i,B)=0$ and the strict comparison carries no parity information about
indices with $x_a>x_i$.

Section~\ref{sec:simulate} simulates one parity test on any~$B$ using
$O(\log\card{B})$ auxiliary polynomials: squared subset sums locate
$\mu_B:=\card{\{a\in B:x_a=x_i\}}$, then one unsquared sum at the critical
subset size collapses to $\prod_{a\in B,\,x_a\neq x_i}(x_i-x_a)$.
Section~\ref{sec:embed} substitutes this into~\cite{TY94}, giving depth
$O((\log n)^3)$ on arbitrary inputs with the same $O(n^{-c})$ error bound.

On the deterministic side, Theorem~7 of~\cite{Rabin72} implies that
every polynomial decision tree for maximum-finding among $n$ distinct reals
uses at least $n-1$ polynomial tests in the worst case; a linear elimination
tree has depth~$n-1$, so the bound is sharp.

\section{Simulating a parity test}\label{sec:simulate}

Fix $i\in[n]$ and $B\subseteq[n]\setminus\{i\}$, and write $b:=\card{B}$,
\[
  \mu_B\;:=\;\card{\{a\in B : x_a=x_i\}},
  \qquad
  B^{\neq}\;:=\;\{a\in B : x_a\neq x_i\},
\]
so $\card{B^{\neq}}=b-\mu_B$.  Note that $P(i,B)=0$ if and only if $\mu_B>0$.

\subsection{Step 1: counting ties via squared sums}

Set $R_{i,B}(0):=0$ and, for $t\in\{1,\dots,b\}$,
\begin{equation}\label{eq:Rm}
  R_{i,B}(t)\;:=\;\sum_{\substack{D\subseteq B\\ \card{D}=b-t+1}}P(i,D)^2.
\end{equation}

\begin{lemma}\label{lem:count-B}
$R_{i,B}(t)=0$ if and only if $\mu_B\ge t$.
\end{lemma}

\begin{proof}
A summand $P(i,D)^2$ vanishes if $D$ meets the $\mu_B$ tied indices, and is
strictly positive otherwise.  If $\mu_B\ge t$, every $D$ of size $b-t+1$ must
meet those indices, since $\card{B\setminus D}=t-1<\mu_B$; hence
$R_{i,B}(t)=0$.  If $\mu_B<t$, take $C\subseteq B$ of size $t-1$ containing
every tied index and set $D:=B\setminus C$; then $\card{D}=b-t+1$, no factor
in $P(i,D)$ vanishes, and $R_{i,B}(t)>0$.
\end{proof}

\begin{corollary}\label{cor:mu}
$\mu_B$ is the largest $t\in\{0,1,\dots,b\}$ with $R_{i,B}(t)=0$, and is found
by binary search using at most $\ceil{\log_2(b+1)}$ polynomial tests
$R_{i,B}(t):0$.
\end{corollary}

\subsection{Step 2: parity after removing the tied indices}

With $\mu_B$ known, set
\begin{equation}\label{eq:Pi}
  \Pi_{i,B}\;:=\;\sum_{\substack{D\subseteq B\\ \card{D}=b-\mu_B}}\,
  \prod_{a\in D}(x_i-x_a).
\end{equation}

\begin{lemma}\label{lem:parity-B}
$\Pi_{i,B}=\prod_{a\in B^{\neq}}(x_i-x_a)$, the empty product being~$1$.
\end{lemma}

\begin{proof}
Among the subsets $D\subseteq B$ of size $b-\mu_B$, the unique one disjoint
from the $\mu_B$ tied indices is $D=B^{\neq}$; every other $D$ contributes a
product with a vanishing factor.
\end{proof}

\begin{proposition}\label{prop:simulate}
For every $B\subseteq[n]\setminus\{i\}$, the sign of $\Pi_{i,B}$ equals the
sign of $P(i,B^{\neq})$ and is determined by at most
$\ceil{\log_2(\card{B}+1)}+1$ polynomial tests of degree at most $2\card{B}$.
\end{proposition}

\begin{proof}
Combine Corollary~\ref{cor:mu} and Lemma~\ref{lem:parity-B}.
\end{proof}

\begin{remark}[Beyond max-finding]
Proposition~\ref{prop:simulate} is a standalone gadget: for any pivot $i$
and set $B$, it turns a parity test $P(i,B):0$ into $O(\log\card{B})$
polynomial tests whose outcome ignores indices tied with $x_i$.
\end{remark}

\section{Substitution into the Ting--Yao algorithm}\label{sec:embed}

Throughout this section, $J_i:=\{\,j\in[n]:x_j>x_i\,\}$.

\subsection{The Ting--Yao algorithm}

Fix $c>0$ and set $m:=\ceil{3c\log_2 n}$, $m_0:=\floor{10c\log_2 n}$.  The
proof of~\cite{TY94} chains three lemmas:
\begin{itemize}
  \item \emph{Lemma~2.}\ \ Draw $m$ independent uniform random subsets
        $B_1,\dots,B_m\subseteq[n]\setminus\{i\}$ and ask the parity tests
        $P(i,B_k):0$.  If every answer is ``$>$'', declare ``$i$ is max'';
        otherwise return the first $B_k$ whose answer is ``$<$''.
  \item \emph{Lemma~3.}\ \ Given $B$ with $P(i,B)<0$, recursively bisect
        $B$, using one parity test per level to choose the half to recurse
        into.  After $\ceil{\log_2\card{B}}$ tests this returns some
        $i'\in B$ with $x_{i'}>x_i$, uniformly distributed over $B\cap J_i$.
  \item \emph{Lemma~1.}\ \ Chain the two; one stage costs at most
        $m+\ceil{\log_2(n-1)}$ parity tests.
\end{itemize}
The outer procedure FINDMAX iterates Lemma~1 for at most $m_0$ rounds, sets
$i\gets i'$ each round, and halts on a declaration or on exhaustion.  In total,
$O((\log n)^2)$ parity tests are made; under distinctness, Theorem~1
of~\cite{TY94} bounds the failure probability by $O(n^{-c})$.

\subsection{Plugging in the simulator}

Replace each parity test $P(i,B):0$ in Lemmas~2 and~3 by the simulator of
Section~\ref{sec:simulate}.  By Proposition~\ref{prop:simulate} this returns
the sign of $\Pi_{i,B}=P(i,B^{\neq})$ in
$O(\log\card{B})\subseteq O(\log n)$ ordinary polynomial tests.  Since
$J_i\cap B=J_i\cap B^{\neq}$,
\begin{equation}\label{eq:sign-parity}
  \sign\Pi_{i,B}\;=\;(-1)^{\card{B\cap J_i}},
\end{equation}
so the simulated answer is ``$>$'' or ``$<$'' according as $\card{B\cap J_i}$
is even or odd; ``$=$'' never occurs.

\paragraph{Lemma~2 with the simulator.}
For each $j\in J_i$ and each $k$, $\Pr[j\in B_k]=\tfrac12$ independently.  If
$x_i$ is a maximum then $J_i=\emptyset$, every simulated answer is ``$>$''
by~\eqref{eq:sign-parity}, and the algorithm declares.  Otherwise
$\card{B_k\cap J_i}$ is odd with probability exactly $\tfrac12$, so the
algorithm outputs some $B_k$ with probability $1-2^{-m}$.  Conditioned on
this event, $B_k\cap J_i$ is a uniform random odd-sized subset of $J_i$.

\paragraph{Lemma~3 with the simulator.}
Maintain the invariant ``$\card{B\cap J_i}$ is odd'', supplied by Lemma~2's
output.  Splitting $B=B'\sqcup B''$, exactly one of $\card{B'\cap J_i}$ and
$\card{B''\cap J_i}$ is odd; by~\eqref{eq:sign-parity} the simulator on $B'$
identifies which, and we recurse into that half.  After $\ceil{\log_2\card{B}}$
simulated tests the recursion reaches a singleton $\{i'\}\subseteq B\cap J_i$,
so $x_{i'}>x_i$.  Each split depends on $B$ only through
$\card{B\cap J_i}\bmod 2$, exactly as in~\cite{TY94}, so their symmetry
argument carries over and $i'$ is uniform on $B\cap J_i$.

\begin{remark}[Exponential search for $\mu_B$]\label{rem:exp-search}
Corollary~\ref{cor:mu} uses $O(\log b)$ tests.  Exponential search on~$t$
with the predicate $R_{i,B}(t)=0$---double $t$ until the predicate fails,
then binary search on the bracketing dyadic interval---uses
$O(\log(\mu_B+1))$ tests, matching the $O(\log b)$ bound only when
$\mu_B=b^{\Theta(1)}$.  The refinement does not change the per-call $O(\log n)$
bound, but amortizing it along the recursion of~\cite{TY94} could in
principle shave one $\log n$ factor from our $O((\log n)^3)$ depth; we do
not pursue this.
\end{remark}

\begin{remark}[A balanced multiset]\label{rem:sqrtn}
A representative tie-heavy input has $\Theta(\sqrt{n})$ distinct values,
each appearing $\Theta(\sqrt{n})$ times (adjust one multiplicity if $n$ is
not a perfect square).
\end{remark}

\subsection{Main result}\label{sec:main}

\begin{theorem}\label{thm:main}
For every fixed $c>0$ there is a randomized polynomial decision tree of depth
$O((\log n)^3)$ that, on any input $x_1,\dots,x_n\in\R$, outputs an index
$i^{\star}\in[n]$ with $x_{i^{\star}}=\max_j x_j$ and failure probability
$O(n^{-c})$.
\end{theorem}

\begin{proof}
\emph{Cost.}  Each of the $O((\log n)^2)$ parity tests of FINDMAX is replaced
by $O(\log n)$ ordinary polynomial tests, for a total of $O((\log n)^3)$.

\emph{Correctness.}  By the analysis of Section~\ref{sec:embed}, the
simulated Lemmas~2 and~3 satisfy the same input/output guarantees as
in~\cite{TY94}: declare when $x_i$ is a maximum, otherwise output $i'$
uniform on $J_i$ with success probability $1-2^{-m}$ per round.  Set
$r_\ell:=\card{J_{i_\ell}}$ along the FINDMAX walk; then $r_\ell<r_{\ell-1}$
at every step, and
\[
  \Pr\bigl[r_\ell>(r_{\ell-1}-1)/2\bigm|i_{\ell-1}\bigr]
  \;=\;\frac{\card{\{j\in J_{i_{\ell-1}}:\card{J_j}>(r_{\ell-1}-1)/2\}}}{r_{\ell-1}}
  \;\le\;\frac{\floor{r_{\ell-1}/2}}{r_{\ell-1}}
  \;\le\;\tfrac12,
\]
by counting positions in the value-sorted list of $J_{i_{\ell-1}}$ (ties
only shrink the numerator).  The rest of~\cite{TY94}'s analysis applies,
giving failure probability $O(n^{-c})$.
\end{proof}

\section*{Acknowledgments}
I thank Michael Ben-Or for bringing this problem and its solution
to my attention.

\bibliographystyle{alpha}
\bibliography{extension}

@article{TY94,
  author  = {Ting, Hing F. and Yao, Andrew C.},
  title   = {{A randomized algorithm for finding maximum with $O((\log n)^2)$ polynomial tests}},
  journal = {Information Processing Letters},
  volume  = {49},
  number  = {2},
  pages   = {165--167},
  year    = {1994},
  doi     = {10.1016/0020-0190(94)90052-3},
}

@article{Rabin72,
  author  = {Rabin, Michael O.},
  title   = {Proving simultaneous positivity of linear forms},
  journal = {Journal of Computer and System Sciences},
  volume  = {6},
  pages   = {639--650},
  year    = {1972},
}

\end{document}